\documentclass[a4paper,bibtotoc,chapterprefix,11pt]{article}

\usepackage{graphicx,color}

\usepackage{amsmath}
\usepackage{amsfonts}
\usepackage{amssymb}
\usepackage{amsthm}

\usepackage{a4wide}

\usepackage[mathscr]{eucal}

\usepackage{index}

\usepackage{fancyhdr}

\usepackage{times}

\newtheorem{thm}{Theorem}[section]

\newtheorem{lemma}[thm]{Lemma}

\newtheorem{fact}[thm]{Fact}

\newcommand{\nc}[1]{\expandafter\newcommand\csname#1\endcsname}

\nc{gX}[1]{\nc{g#1}{\mathfrak #1}} %group X

\nc{bX}[1]{\nc{b#1}{\mathbb #1}} %group X

\nc{aX}[1]{\nc{a#1}{\mathcal #1}} %additive analog X
\gX G
\gX H
\gX C
\gX S
\gX A

\bX A

\aX H
\aX G
\aX S
 \newlength{\ppshort}
\usepackage{mathrsfs}

\newcommand{\E}{\mathbb{E}}

\newcommand{\El}{\E_{\textsf{lin}}}

\newcommand{\lA}{\mathcal A}

\newcommand{\constone}{\mathbf{1}}

\newcommand{\gft}{\textup{GF}(2)}

\begin{document}
%% amsart title:
%\title{An Approximation Algorithm for $\#k$-SAT}
%\author{Marc Thurley}
%\address{Centre de Recerca Matem\`{a}tica\\
%Bellaterra, Spain}
%\thanks{The author was supported by Marie Curie Intra-European Fellowship 271959 at the Centre de Recerca Matem\`{a}tica, Bellaterra, Spain}
%\email{marc.thurley@googlemail.com}

\title{An Approximation Algorithm for $\#k$-SAT}
\author{Marc Thurley\footnote{supported by Marie Curie Intra-European Fellowship 271959 at the Centre de Recerca Matem\`{a}tica, Bellaterra, Spain}
\\
Centre de Recerca Matem\`{a}tica\\
Bellaterra, Spain}
%\email{marc.thurley@googlemail.com}

%\footnote{Centre de Recerca Matem\`{a}tica, Bellaterra, Spain, supported in part 
%by Marie Curie Intra-European Fellowship 271959}}
\maketitle
\begin{abstract}
We present a simple randomized algorithm that approximates the number
of satisfying assignments of Boolean formulas in conjunctive normal form.
To the best of our knowledge this is the first algorithm 
which approximates $\#k$-SAT for any $k\ge 3$ within
a running time that is not only non-trivial, but also significantly better
than that of the currently fastest exact algorithms for the problem.
More precisely, our algorithm is a randomized approximation scheme
whose running time depends polynomially on the error tolerance 
%$\epsilon > 0$
and is mildly exponential in the number $n$ of variables of the
input formula. For example, even stipulating sub-exponentially small error tolerance, 
the number of solutions to $3$-CNF input formulas can be approximated 
in time $O(1.5366^n)$. For $4$-CNF input the bound increases to $O(1.6155^n)$. 

We further show how to obtain upper and lower bounds on the number 
of solutions to a CNF formula in a controllable way. Relaxing the
requirements on the quality of the approximation, on $k$-CNF input
we obtain significantly
reduced running times in comparison to the above bounds.
\end{abstract}

\pagestyle{plain}

\newcommand{\sat}{\textup{sat}}
\newcommand{\forced}{\textup{Forced}}
\newcommand{\fplin}{\textup{F}_{\oplus}}

\section{Introduction}
The design and analysis of algorithms that determine the
satisfiability or count the models of $k$-CNF formulas
has quite some tradition. In the case of the satisfiability 
problem the earliest algorithm with a worst case running time 
which is significantly better than the trivial $\text{poly}(n)2^n$ 
bound dates back to at least 1985 \cite{monspe85}. The time bounds
have improved gradually over the years with most recent results
(only a few of which are \cite{iwatam04, rol06, hermossch11,her11})
being analyses of randomized algorithms that 
have been obtained from either Sch\"oning`s
algorithm \cite{sch99}, the algorithm of Paturi, Pudl\'{a}k, Saks, and Zane~\cite{patpudsaksan05},
or a combination of both. 
The currently fastest algorithm for $3$-SAT by Hertli \cite{her11}
running in time $O(1.30704^n)$ falls roughly into the second category.
The corresponding counting problems have seen similar improvements 
\cite{dub91,zha96a,dahjonwah05,kut07,furkas07,wah08}
over the trivial time bound, with the current best worst case running time
for $\#3$-SAT being $O(1.6423^n)$ obtained by Kutzkov~\cite{kut07}.

Quite surprisingly, however, the situation is completely different
for the \emph{approximation} of $\#k$-SAT. To the best of my knowledge 
not even small improvements over the trivial worst case time bound are 
known.\footnote{Disregarding, of course the pathological fact that
exact methods can be interpreted as approximation algorithms, as well.}
This, however, does not seem to be due to a general lack of 
interest in the problem itself.
From a complexity theoretical point of view, for example,
several already classic papers 
\cite{sto76,valvaz86,rot96} study questions closely related
to direct algorithmic problems in $\#k$-SAT approximation.
In particular Valiant and Vazirani~\cite{valvaz86} bound the
complexity of the approximation problem from above 
by reduction to SAT, and hence settle its complexity in a certain sense.

While theoretical results on the approximation of $\#k$-SAT are rather old,
there are several heuristic approaches to the problem,
that have all appeared only fairly recently.
Motivated by questions of practicability, these results
focus on methods
that can be shown empirically to work well, 
while sacrificing some (at least theoretically) desirable
properties.
That is, some of these approaches yield approximations without any guarantee
on the quality of the approximation \cite{weisel05,gogdec07}.
Others yield reliable lower and upper bounds \cite{gomsabsel06,gomhofsabsel07,krosabsel11}
which, in certain cases, are surprisingly
good although generally their quality is rather low. 
In particular, this line of work does not provide rigorous bounds on running times
and neither does it yield rigorous quality estimates of the approximation computed.

With regard to the above results, the lack of competitive worst case bounds for
$\#k$-SAT approximation algorithms seems to be due to several factors.
First of all the exact algorithms found in the literature and their analyses
do not seem to carry over easily to the approximation problem. Secondly,
complexity theoretical insights are usually not considered applicable in the context of
designing fast exponential algorithms.
An example is the technique of Valiant and Vazirani
which leads to a significant blow up in formula size. 
And thirdly, it is not clear which of the known algorithmic ideas 
used in the heuristic approaches could at least in principle show a
good worst case behavior.

\subsection{Contributions}
In this paper we will see that one can indeed not only improve upon the trivial
worst-case time bound mentioned above. But the algorithm we will present
also provides arbitrarily good precision in significantly less time than known exact methods.
To be more precise, the algorithm we present 
is a \emph{randomized approximation scheme} for $\#k$-SAT for every $k \ge 3$.
Given a freely adjustable error tolerance $\epsilon > 0$, randomized 
approximation schemes produce an output 
that is within a factor of $e^\epsilon$ of the number $\#F$ of solutions of
some input formula $F$.

We obtain the following main result, which we state here 
only for $k=3$ and $k=4$. The general result will be discussed in Section
\ref{sec:algo}. 
\begin{thm} \label{thm:main_popular}
There is a randomized approximation scheme running in time
$O(\epsilon^{-2}\cdot 1.5366^{n})$ for $\#3$-SAT  and in time $O(\epsilon^{-2} \cdot 1.6155^{n})$ for $\#4$-SAT.
\end{thm}

For $\#3$-SAT this algorithm is already significantly faster than the 
currently fastest
exact algorithm from~\cite{kut07} which runs in time $O(1.6423^n)$.
For $\#4$-SAT the benefit of approximation is even more impressive, as the
best bound for exact methods is still the $O(1.9275^n)$ bound of 
the basically identical algorithms of Dubois~\cite{dub91} and Zhang~\cite{zha96a}.

We will see that the algorithm of Theorem~\ref{thm:main_popular}
is not complicated and monolithic like the branching algorithms usually
employed in exact counting results. But it is actually 
a combination of two very simple and very different algorithms.
The main reason for considering this combination relies on two pieces of
intuition.
On the one hand, if a formula has few solutions, then it is not too bad an idea to
compute their number by simply enumerating them.
On the other hand, if a formula has many solutions, then
a quite trivial sampling algorithm should yield good results. 

Observe that the result of Theorem~\ref{thm:main_popular} 
can already be used to compute e.g. $\#F$ exactly in time $O(1.5366^{n})$
for any $3$-CNF formula which has only a sub-exponential number of solutions.
To achieve this we only have to set $\epsilon$ appropriately.
However, we shall see below, that this can also be achieved 
in significantly less time.
Motivated by the heuristic results on the approximation of 
$\#k$-SAT described above, we also study the effect of
weaker requirements on the 
approximation bounds.
It seems, of course, perfectly reasonable to assume that weaker bounds 
should come at the benefit of dramatically improved
running time bounds. We will therefore show that this is the case.
With respect to lower bounds we obtain:
\begin{thm}[Lower Bound Algorithm]\label{thm:lower}
There is a randomized algorithm which, on input a $3$-CNF formula $F$ on $n$ variables
and a natural number $L$, performs 
the following in time $O(L^{0.614}\cdot 1.30704^n)$:
\begin{itemize}
 \item If $\#F > L$ it reports this with probability at least $3/4$.
 \item If $\#F \le L$ then with probability at least $3/4$ it reports this and outputs the correct value $\#F$. 
\end{itemize}
Furthermore, there is a deterministic algorithm solving this task in time
$O(N^{0.585}\cdot 1.3334^n)$.
\end{thm}

This lower bound algorithm will in fact be used in the proof of Theorem~\ref{thm:main_popular}
and relies on the above observation that we can simply use a SAT algorithm for enumerating all solutions
provided the input formula has only few. The time bounds mentioned thus arise from the SAT algorithms used
-- the randomized $3$-SAT algorithm by Hertli~\cite{her11} and the
deterministic one of Moser and Scheder~\cite{mossch11} (which is in fact a derandomized version of Sch\"oning's
algorithm). 

To obtain upper bounds, on the other hand,
we cannot use the high-solution part of Theorem~\ref{thm:main_popular}.
But,
although it might seem unreasonable to expect that this would yield a competitive running time, 
we can use an algorithm based on the bisection technique of Valiant and Vazirani~\cite{valvaz86}.
Interestingly an algorithm based on Valiant and Vazirani`s technique has been used already in the
heuristic result of~\cite{gomsabsel06}. Their approach, however, is quite different from ours and 
does not have a good worst-case behavior. 
 
By systematically 
augmenting and input formula $F$ with randomly chosen $\textup{GF}(2)$-linear constraints, 
the bisection technique
makes it possible to approximate $\#F$ by determining satisfiability of the augmented formulas.
The main difference of our approach to this classical scheme lies
in the observation that it is more reasonable
for our purposes to work directly with the system of linear equations obtained, instead of encoding
it into $k$-CNF. In this way we obtain the running time bounds which are valid even for general 
CNF input formulas.

\begin{thm}[Upper Bound Algorithm]\label{thm:approx_many_solutions}
There is an algorithm, which on input a CNF formula $F$ on $n$ variables
and an integer $\mu \le n$ takes time $O^*(2^{n-\mu})$
and performs the following with probability at least $2/3$:

It outputs a number $u \ge \mu$ such that $U := 2^{u + 3} \ge \#F$.
If furthermore $u > \mu$ then $2^{u}$ is a $16$-approximation of $\#F$.
\end{thm}

\paragraph*{Remark.} 
This algorithm will actually work for upper bounding
$|S|$ for any set $S \subseteq \{0,1\}^n$ with a polynomial membership test. 
However, as this is a trivial
consequence of the proof, we consider only the case that 
the input is a CNF formula.

Moreover, we do not particularly focus on improving the approximation
ratio mentioned in Theorem~\ref{thm:approx_many_solutions}. 
Such an improvement is in fact unnecessary if we want to use this
algorithm to design a randomized approximation scheme:
We can combine the above algorithm with that of
Theorem~\ref{thm:lower} to obtain a $16$-approximation algorithm
for e.g. $\#3$-SAT which runs (up to a polynomial factor) within in the same time bound
as that stated in Theorem~\ref{thm:main_popular}.  This algorithm can then be plugged into 
a Markov chain by Jerrum and Sinclair \cite{sinjer89}
to boost the quality of approximation.
This yields a $(1 + \frac{1}{\textup{poly}(n)})$-approximation
algorithm incurring only a polynomial overhead in the computation.
Thus we have a second, although more complicated, algorithm that satisfies the claim of Theorem~\ref{thm:main_popular}.

\section{Preliminaries}
For a CNF formula $F$, let $\sat(F)$ be the set of its solutions and
$\#F = |\sat(F)|$. We shall always use $n$ to denote the number of variables of a
CNF formula under consideration. 
A \emph{randomized $\alpha$-approximation algorithm} $\bA$ for $\#k$-SAT 
outputs, on input a $k$-CNF formula $F$, a number $\bA(F)$ such that
\begin{equation}\label{eq:aapprox}
\Pr\left[\alpha^{-1}\#F \le \bA(F) \le \alpha \#F\right] \ge p.
\end{equation}
Where $p$ is some constant, independent of the input and strictly larger 
than\footnote{In the literature, usually either the value $p = 3/4$ or 
a further parameter $\delta$ such that $p = 1 - \delta$ seems to be common.
However, it is well-known that all of these can be translated into each other
with only polynomial overhead.}
$1/2$.
A \emph{randomized approximation scheme} for $\#k$-SAT, is then an algorithm 
which on input $F$ and a natural number $\epsilon^{-1}$ 
behaves like a randomized $e^\epsilon$-approximation algorithm.

We use the notation $x^1 = x$ and $x^0 = \bar x$. 
For a clause $C$, a variable $x$, and a truth
value $a \in \{0,1\}$, the \emph{restriction} of $C$ on $x=a$ is the
constant $\constone$ if the literal $x^a$ belongs to $C$, and 
$C \setminus \{ x^{1-a} \}$ otherwise. 
We write $C|_{x=a}$ for the restriction of $C$
on $x=a$.
A \emph{partial assignment} is a sequence of assignments $( x_1 = a_1,
\ldots, x_r = a_r )$ with all variables distinct.  Let $\alpha$ be a
partial assignment. We will use the notation $\alpha \cup (x= a)$ 
to denote the assignment $( x_1 = a_1,
\ldots, x_r = a_r, x = a)$.
If $C$ is a clause, we let $C|_\alpha$ be the
result of applying the restrictions $x_1 = a_1, \ldots, x_r = a_r$ to
$C$. Clearly the order of application does not matter. 
If $F$ is a CNF formula, we let
$F|_\alpha$ denote the result of applying the restriction $\alpha$ to
each clause in $F$, and removing the resulting $\constone$'s. We call
$F|_\alpha$ the \emph{residual} formula.

As we will use the algorithm of Paturi, Pudl\'{a}k, Saks, and Zane~\cite{patpudsaksan05} 
and a very recent paper by Hertli \cite{her11},
we need the constant
$$
\mu_k = \sum_{j=1}^\infty \frac{1}{j(j + \frac{1}{k-1})}.
$$
\section{The Algorithm}\label{sec:algo}
We are now able to state the main result in full detail.
\begin{thm}\label{thm:main}
For $k \ge 3$, $\#k$-SAT has a randomized approximation scheme running in 
time\footnote{We use the $O^*$ notation to suppress factors sub-exponential in $n$.}
$$
O^*\left(\epsilon^{-2} \cdot 2^{n\left(\frac{k-1}{k-1+\mu_k}\right)}\right).
$$
\end{thm}
As already outlined, the randomized approximation scheme of Theorem~\ref{thm:main} 
is a combination of two different algorithms. We will discuss the algorithm for the
case of few solutions now. The case of many solutions will be treated 
afterwards in Section~\ref{sec:many_sols}.

\subsection{Formulas with few solutions}
For formulas with few solutions we will now present an algorithm relying
on a simple enumeration of solutions by using a $k$-SAT algorithm as a subroutine.
This will also prove Theorem~\ref{thm:lower}.
\begin{lemma}\label{lem:few_sols}
Let $F$ be a $k$-CNF formula on $n$ variables
and let $\bA$ be an algorithm performing the following task  
in time $O^*(2^{\beta_k n})$. 
If $F$ is satisfiable, with probability at least $3/4$
it outputs a solution to $F$. If $F$ is unsatisfiable,
it reports this correctly.
 
Then, there is a algorithm, which on input $F$ and a natural number $N$, takes time
$O^{*}(2^{\beta_k n} N^{(1-\beta_k)})$,
and performs the following:
\begin{itemize}
 \item If $\#F > N$ it reports this with probability at least $3/4$.
 \item If $\#F \le N$ then with probability at least $3/4$ it reports this and outputs the correct value $\#F$. 
\end{itemize}
Furthermore, if the algorithm reports $\#F > N$ then this holds with certainty.
\end{lemma}
Theorem~\ref{thm:lower} follows directly from this lemma 
by using the randomized $3$-SAT algorithm of Hertli~\cite{her11}
which has $\beta_3 = 0.3864$.
For the claim about the deterministic algorithm
we use the result of Moser and Scheder~\cite{mossch11},
with $\beta_3 = 0.4151$.

In the proof of the above lemma, we will use the 
following fact which is very easily proven. % by induction on $N$,
\begin{lemma}\label{lem:2906111400}
A rooted tree with $N$ leaves and depth (i.e. max root to leaf-distance) $n$
has at most $n\cdot N$ vertices in total. 
\end{lemma}

\begin{proof}[Proof of Lemma~\ref{lem:few_sols}]
Note first, that by a standard trick we can boost the success probability
of $\bA$. Assume that, as provided by the statement of the lemma,
we have error probability at most $1- p \le 1/4$.
Then the probability of erring in $M$ independent repetitions is at
most $(1-p)^M \le e^{-pM}$. Call the boosted version of this algorithm $\bA^*$.

We shall fix a good value for $M$. Below we will see that
algorithm $\bA^*$ will be queried a number $O(nN)$ of times,
for some $N \le 2^n$,
each time on a formula of at most $n$ variables. The probability that the algorithm 
errs in any of these queries, is at most
$nN \cdot e^{-pM}$.
So choosing $M$ within a constant factor of $\log nN$ (which is polynomial in $n$) 
allows us to 
condition on the SAT algorithm not erring in any of the $O(nN)$ queries.
The probability of that latter event to happen is close to $1$.
And as this is the only possible source of failure of the algorithm,
we will easily achieve a success probability of $3/4$ in the end.

\paragraph*{The algorithm.}

Check if $F$ is satisfiable, using $\bA^*$, and if so,
perform the following. Inductively, construct a search tree associated with partial assignments $\alpha$,
such that $F\vert_\alpha$ is satisfiable.
For a leaf in the current search tree associated with some assignment $\alpha$, 
choose a variable $x$ from $F\vert_\alpha$
and check $F\vert_{\alpha \cup (x=0)}$ and $F\vert_{\alpha \cup (x=1)}$ for satisfiability
using the algorithm $\bA^*$.
For each of the satisfiable restrictions add a new child to the current leaf 
in the search tree.
We stop the algorithm, if it has $N$ leaves, or if it has found all of the
solutions of $F$.
Traversing this tree, e.g. in a depth first manner we can implement this procedure
in polynomial space (not taking the space need of $\bA^*$ into account).

\paragraph*{Time.}
Consider the search tree this algorithm produces.
As it has at most $N$ leaves, and depth at most $n$, we have (recall Lemma~\ref{lem:2906111400})
at most $nN$ nodes overall in the tree.

Observe first, that at each node we perform at most $2$ queries to the SAT algorithm $\bA^*$.
A node of level $\ell$ in the tree incurs queries taking time at most $O^*(b^{n-\ell})$ for $b = 2^{\beta_k}$.
Therefore, we can give an upper bound of the overall time spent in answering all queries
by bounding the time spend on a completely balanced binary search tree of depth $d = \log nN$.

Let $T(d,n)$ denote the overall time spend to run the algorithm on a balanced binary search tree with $d$ 
levels with an $n$ variable formula.
Then, up to a sub-exponential factor for the time spent at each node in the tree, $T(d,n) = b^{n} + 2T(d-1,n-1)$.
Note that $T(0,n) = 1$, and thus
$$
T(d,n) = \sum_{\nu = 0}^d 2^\nu b^{n-\nu}
$$
which yields the claimed bound.
\end{proof}

\subsection{Formulas with many Solutions.}\label{sec:many_sols}
We use the following simple folklore algorithm which can be found
e.g. in Motwani and Raghavan's book \cite{motrag95}.
Given a CNF formula $F$ on $n$ variables, 
choose an assignment from $\{0,1\}^n$ uniformly at random.
Repeat this process a number $N$ of times and let $X$
be the number of solutions of $F$ among these $N$ trials.
By a simple argument (see e.g. Theorem 11.1. in \cite{motrag95}),
if
$$
N = \Omega\left(\dfrac{2^n}{\epsilon^2 \#F}\right)
$$
then with probability at least $3/4$, 
the value $X\cdot \dfrac{2^n}{N}$ is an $e^\epsilon$-approximation
of $\#F$. 
%This is within a small (actually polynomial, for $\epsilon$ inverse polynomial)
%factor of
%$$
%N \ge 2^{n(1-f)}.
%$$
Hence, we have the following
\begin{lemma}\label{lem:many_sols}
Let $F$ be an $n$ variable CNF formula with at least $N$ solutions
and $\epsilon^{-1}$ a natural number.
Then there is an algorithm which, in time $O^*\left(\dfrac{2^{n}}{\epsilon^{2}N}\right)$
yields a randomized $e^\epsilon$-approximation of $\#F$.
\end{lemma}

\subsection{Combining the algorithms}
We shall now prove Theorem~\ref{thm:main} by combining both of the above algorithms.
Let $F$ be a $k$-CNF formula on $n$ variables and $\epsilon^{-1}$ a natural number.
We run the algorithm of Lemma~\ref{lem:few_sols} with a parameter $N$.
The exact value of $N$ will be determined later.
Note that if the algorithm reports $\#F \le N$, it also computes $\#F$ exactly with probability at least $3/4$.
Otherwise, if the algorithm reports that $\#F > N$, we know with certainty that this is the case.
Hence, given that the algorithm reports the latter, the algorithm of Lemma~\ref{lem:many_sols} will take time
$O^*\left(\dfrac{2^{n}}{\epsilon^{2}N}\right)$ to yield an $e^\epsilon$-approximation of $\#F$.

It remains to bound the running time which amounts to optimizing the cutoff parameter $N$. 
For every choice of $N$, the combined algorithm works in time within a sub-exponential factor of
$$
\max \{2^{\beta_k n} N^{(1-\beta_k)},\dfrac{2^{n}}{N}\}.
$$
Let $f$ be such that $\log_2 N = f\cdot n$.
Then this maximum translates into $\max \{\beta_k + f (1-\beta_k), 1-f\}$.
Since $(\beta_k  + f(1-\beta_k)$ is increasing and $1-f$ is decreasing in $f$,
the minimum over all $f$ of the maximum of the two is obtained when $f$ is chosen so as to make them 
equal, that is, at
$$
f = \dfrac{1-\beta_k}{2-\beta_k}.
$$
This translates into an overall running time of
$$
O^*\left(2^{\frac{n}{2-\beta_k}}\right).
$$
Recall that $\beta_k$ determines the running time $O^*(2^{\beta_kn})$ of the subroutine consisting of a
randomized $k$-SAT algorithm, used in the algorithm of Lemma~\ref{lem:few_sols}. 
We shall have a look at these running times, now.
\paragraph{The case $k \ge 5$.}
The algorithm of Paturi, Pudl\'{a}k, Saks and Zane \cite{patpudsaksan05},
can be used as the subroutine randomized $k$-SAT algorithm,
which has a running time of
\begin{equation}\label{eq:ppsz_time}
O^*\left(2^{(1-\frac{\mu_k}{k-1})n}\right).
\end{equation}
Hence, we have here, $\beta_k = 1 - \frac{\mu_k}{k-1}$ which yields the claimed bound. 

\paragraph{The cases $k = 3$ and $k=4$.}
For these values of $k$, several improvements over the PPSZ algorithm have been presented.
The currently fastest one is that by Hertli \cite{her11}, 
whose bounds match those of the PPSZ algorithm in the unique-SAT case.
We thus also have here the corresponding bound of equation \eqref{eq:ppsz_time}.
\section{Upper bounds}
In this section we will prove Theorem~\ref{thm:approx_many_solutions}
by presenting a simple algorithm producing upper bounds on
$\#F$. We will use Valiant and Vazirani's 
bisection technique~\cite{valvaz86} 
and its application to approximate counting.
We will therefore consider random $\gft$-linear systems of equations of the form $A x = b$. 
For some $m \le n$
these consist of an $m \times n$ matrix $A$, an $m$ dimensional 
vector $b$ and a vector $x$ representing the variables of $F$. The entries of $A$ and $b$ 
are chosen independently and uniformly at random from $\{0,1\}$.
As such systems give rise to a family of pairwise independent hash functions of the form
$h(x) = Ax - b$, we will use the following 
well known
\begin{fact}[Hashing Lemma]
\label{lem:hash}
 Let $F$ be a CNF formula, and $A x = b$ an $m \times n$ random linear system of equations. Then 
 $$
  \Pr\left[\left|\{x \in \sat(F) \mid A x = b\}\right| \notin (1-\epsilon,1+\epsilon)\cdot 2^{-m}\cdot \#F\right] 
  \le \dfrac{(1-2^{-m})2^m}{\#F \cdot \epsilon^2}.
 $$
\end{fact}
The proof of this lemma is straightforward and can be found e.g. in Goldreich's book \cite{gol08}.
Secondly, we need a standard fact about the rank of random matrices such as the matrices obtained
in the above way. Consider a random $m \times n$ matrix $A$ as above with $m \le n$ and let $r$ 
denote its rank. The proof of the following lemma then follows easily, for example, from a result of
Bl\"omer, Karp and Welzl~\cite{blokarwel97}:
\begin{lemma}\label{lem:rank_lin_sys}
There is a constant $c$ such that $\El[r] \ge m-c$. 
Furthermore, $r \ge m - O(\log m)$ with probability at least $1 - O(m^{-1})$.
\end{lemma}
%\begin{proof}
%See Theorem 2.1. in \cite{blokarwel97}: $\El[2^{m-r}] = O(1)$
%this implies that $\El[m-r] \in O(1)$.
%Furthermore, this also implies, by Markov's inequality,
%$$
%\Pr[n-r \ge k] = \Pr[2^{n-r} \ge 2^k] \le \dfrac{\El[2^{n-r}]}{2^k}.
%$$
%Which yields the second claim.
%\end{proof}
%
A third ingredient of the algorithm is the following Lemma, which is easily proved.
\begin{lemma}\label{lem:algo_all_linsols}
Let $Ax = b$ be a system of $\gft$-linear equations with solution set $\lA$.
There is an algorithm listing all solutions in time within a polynomial factor
of $|\lA|$.
\end{lemma}
%\begin{proof}
%This is easy. We first determine satisfiability of $Ax = b$.
%Given it is satisfiable, we fix a basis of the column space of $A$.
%cycling through all possible assignments of non-basis variables,
%each uniquely fixes the values basis variables. 
%\end{proof}
%
We are now ready to prove the Theorem.
\begin{proof}[Proof of Theorem~\ref{thm:approx_many_solutions}]
We start with the description of the algorithm.
Choose a random $\gft$-linear $n \times n$ system of equations $A x = b$.
Starting with a parameter $\nu = n$ and decreasing $\nu$ in each step, we
build random linear systems $A_nx = b_n, A_{n-1}x = b_{n-1},\ldots, A_\mu x = b_\mu$.
The system $A_{\nu}x = b_{\nu}$ is obtained from $A x = b$ by deleting the
last $n-\nu$ rows of $A$ and entries of $b$. Then $F_{\nu}$ denotes
the pair consisting of $F$ and $A_\nu x = b_\nu$. We say that $F_\nu$
is \emph{satisfiable} if there is a solution $x \in \sat(F)$ such that 
$A_\nu x = b_\nu$.

For each $\nu$ we rigorously determine whether
$F_{\nu}$ is satisfiable by using the algorithm of Lemma~\ref{lem:algo_all_linsols}
to list all solutions of the linear system and for each determine whether it satisfies $F$.
We let $u$ be the minimum $\nu \ge \mu$ such that $F_{\nu}$ is unsatisfiable. 
If all $F_{n}, \ldots, F_{\mu}$ are unsatisfiable, we set $u = \mu$.

To establish the time bound, note that the running time is dominated by the
time the algorithm spends in determining satisfiability of $F_{\mu}$.
By Lemma~\ref{lem:rank_lin_sys} the matrix $A_\mu$ has with high probability rank at least 
$m - O(\log n)$ and we have $|\lA| = O^*(2^{n-\mu})$ which yields the claimed time bound.  
Thus we shall in the following condition on the event that the rank of $A_\mu$ satisfies
this rank criterion.

\paragraph*{Correctness.} The correctness follows from standard arguments also used in the
classical approach~\cite{valvaz86}. We give a proof for completeness.
Let $f = \lceil\log \#F \rceil$. Assume first, for simplicity, that $\mu=0$. We will show
that with the desired probability, $u$ is a $16$-approximation to $\#F$.

First, consider the event that $u = f - c$ for some $c \ge 4$.
The probability $p_c$ of this event is
the probability that all of $F_n, \ldots, F_{u}$ are unsatisfiable.
Furthermore, conditional on $F_{u}$ being unsatisfiable, all $F_{\nu}$ for $\nu \ge u$ 
are unsatisfiable, as well. Hence, we have $p_c = \Pr[F_{u} \text{ is unsat }]$, and
by the Hashing Lemma \ref{lem:hash}, we have thus
$$
\Pr[F_{u} \text{ is unsat }] < \dfrac{2^{u}}{\#F} \le 2^{1-c}.
$$
By a union bound argument, we thus see that $f - 3 \le u$ with probability at least $3/4$.

Next, consider $u = f + c$ which implies that $F_{u-1}$ is satisfiable.
Applying the Hashing Lemma~\ref{lem:hash} 
with parameter $\epsilon = \zeta - 1$ for $\zeta = 2^{u-1}/\#F$,
we see that 
$$
\Pr[F_{u-1} \text{ is satisfiable}] < \dfrac{\zeta}{(\zeta - 1)^2}.
$$
As $\zeta \ge 1$, this bound is decreasing in $\zeta$, and hence in $u$,
therefore, the probability that $F_{u-1}$ is satisfiable is at most 
$2^{c-1}(2^{c-1} - 1)^{-2}$. Again, a union bound shows that
$u \le f + 3$ with probability at least $33/49$.

Next, note that if $\mu > 0$ then these findings do not change. Especially,
if $\mu < f - 3$ then the above result does not change. And 
if $\mu \ge f - 3$, then by the above, with probability at least $3/4$ we have $u = \mu$.

Taking into account that we have conditioned on $A_\mu$ having rank $m - O(\log m)$ the above
probabilities degenerate a bit. But the bounds claimed in the statement of the Theorem are still easily 
achieved.
\end{proof}

\paragraph*{Remarks.} Note that the use of listing algorithm of Lemma~\ref{lem:algo_all_linsols}
can be avoided by using a uniform sampling algorithm for the solutions of $A x = b$, this then yields 
essentially the same time bounds. Furthermore, uniform sampling is easily achieved by fixing a basis of
the column space of $A$, choosing u.a.r. assignments to non-basis variables and extending
these assignments to solutions of $Ax = b$.

\section{Open Problems}
It is a peculiar fact that our result falls short of yielding any reasonable time bound for the 
approximation of $\#2$-SAT. A direct application of the algorithm of Theorem~\ref{thm:main}
to this case would yield (using a polynomial time $2$-SAT subroutine) a bound of $O(1.4142^n)$
whereas the fastest exact method \cite{wah08} takes time only $O(1.2377^n)$. It would therefore
be interesting to develop an approximation algorithm which beats the bounds of these exact methods.

Secondly, the time bounds achieved in this paper are significantly better 
than those for known exact methods, but also, they are much worse than the bounds known for
the corresponding satisfiability problems. Is it possible to close this gap, maybe even
in terms of a purely algorithmic analog of Valiant and Vazirani's result?

\section{Acknowledgments}
I would like to thank Martin Grohe for bringing this problem to my attention
and for several helpful discussions on the topic. For further discussions, I would also like to thank
Holger Dell, Alistair Sinclair, and Piyush Srivastava.

\bibliographystyle{amsalpha}
\bibliography{bib}

\end{document}